\documentclass[11pt,a4paper]{article}

\usepackage[margin=1in]{geometry}

\usepackage{theorem,latexsym,graphicx,amssymb}
\usepackage{amsmath,enumerate}
\usepackage{wrapfig}
\usepackage{subfigure}
\usepackage{float}
\usepackage{graphicx}
\usepackage{epsfig}
\usepackage{xspace}
\usepackage{paralist}
\usepackage{times}
\usepackage[compact]{titlesec}
\usepackage{algpseudocode}
\usepackage{algorithm}

\usepackage{enumerate}
\usepackage{cases}
\usepackage{caption}
\usepackage{url}
\usepackage{tikz}
\usetikzlibrary{arrows}
\usetikzlibrary{topaths,calc}
\usepackage{booktabs}
\usepackage{mathtools}

\usepackage{comment}

\usepackage{color}


\usepackage{multicol}

\newenvironment{proof}{{\bf Proof:  }}{}
\newenvironment{proofof}[1]{{\bf Proof of #1:  }}{\hfill \qed}

\numberwithin{figure}{section}
\numberwithin{equation}{section}

\newtheorem{theorem}{Theorem}[section]
\newtheorem{definition}[theorem]{Definition}

\newtheorem{lemma}[theorem]{Lemma}

\newtheorem{proposition}{\hskip\parindent Proposition}[section]

\newcommand{\qed}{\hfill$\Box$}

\newcommand{\defeq}{:=}

\newcommand{\R}{\mathbb{R}}
\newcommand{\E}{\mathbb{E}}

\newcommand{\W}{\ensuremath{\mathsf{W}}\xspace}
\newcommand{\Wm}{\ensuremath{\mathsf{W}^{-1}}\xspace}

\newcommand{\Wh}{\ensuremath{\mathsf{W}^{\frac{1}{2}}}\xspace}

\newcommand{\Wmh}{\ensuremath{\mathsf{W}^{-\frac{1}{2}}}\xspace}

\newcommand{\Lo}{\ensuremath{\mathsf{L}\xspace}}
\newcommand{\Lc}{\ensuremath{\mathcal{L}}\xspace}

\newcommand{\D}{\ensuremath{\mathsf{D}\xspace}}
\newcommand{\Dc}{\ensuremath{\mathcal{D}}\xspace}

\newcommand{\Q}{\ensuremath{\mathsf{Q}\xspace}}

\newcommand{\I}{\ensuremath{\mathsf{I}_n\xspace}}

\newcommand{\w}[1]{\ensuremath{w_{#1}}}
\newcommand{\f}[1]{\ensuremath{f_{#1}}}

\newcommand{\ray}{\ensuremath{\mathsf{R}}\xspace}
\newcommand{\rayc}{\ensuremath{\mathcal{R}}\xspace}

\newcommand{\vp}{\varphi}

\newcommand{\T}{\ensuremath{\mathsf{T}}\xspace}
\newcommand{\ra}{\rightarrow}

\newcommand{\spn}{\operatorname{span}}

\newcommand{\eps}{\epsilon}

\newcounter{note}[section]

\newcommand{\initOneLiners}{%
    \setlength{\itemsep}{0pt}
    \setlength{\parsep }{0pt}
    \setlength{\topsep }{0pt}
}

\newcommand{\ignore}[1]{}

\newcommand*\samethanks[1][\value{footnote}]{\footnotemark[#1]}

\newcommand{\shortv}[1]{}

\usepackage{tabularx}
\usepackage{amsmath,amstext,amssymb,amsfonts}
\usepackage[mathscr]{euscript}

\newcommand{\set}[1]{\left\{#1\right\}}

\newcommand{\Abs}[1]{\left\lvert#1\right\rvert}

\newcommand{\argmax}{{\sf argmax}}
\newcommand{\argmin}{{\sf argmin}}

\newcommand{\supp}{{\sf supp}}

\newcommand{\one}{\ensuremath{\mathbf{1}}} 
\newcommand{\zero}{\ensuremath{\mathbf{0}}}

\newcommand{\eset}{\emptyset}

\newtheorem{theorem*}{Theorem}

\usepackage{boxedminipage}

\usepackage{prettyref}
\newcommand{\savehyperref}[2]{\texorpdfstring{\hyperref[#1]{#2}}{#2}}

\newrefformat{eq}{\savehyperref{#1}{\textup{(\ref*{#1})}}}
\newrefformat{lem}{\savehyperref{#1}{Lemma~\ref*{#1}}}
\newrefformat{def}{\savehyperref{#1}{Definition~\ref*{#1}}}
\newrefformat{thm}{\savehyperref{#1}{Theorem~\ref*{#1}}}
\newrefformat{cor}{\savehyperref{#1}{Corollary~\ref*{#1}}}
\newrefformat{cha}{\savehyperref{#1}{Chapter~\ref*{#1}}}
\newrefformat{sec}{\savehyperref{#1}{Section~\ref*{#1}}}
\newrefformat{app}{\savehyperref{#1}{Appendix~\ref*{#1}}}
\newrefformat{tab}{\savehyperref{#1}{Table~\ref*{#1}}}
\newrefformat{fig}{\savehyperref{#1}{Figure~\ref*{#1}}}
\newrefformat{hyp}{\savehyperref{#1}{Hypothesis~\ref*{#1}}}
\newrefformat{alg}{\savehyperref{#1}{Algorithm~\ref*{#1}}}
\newrefformat{sdp}{\savehyperref{#1}{SDP~\ref*{#1}}}
\newrefformat{rem}{\savehyperref{#1}{Remark~\ref*{#1}}}
\newrefformat{item}{\savehyperref{#1}{Item~\ref*{#1}}}
\newrefformat{step}{\savehyperref{#1}{step~\ref*{#1}}}
\newrefformat{conj}{\savehyperref{#1}{Conjecture~\ref*{#1}}}
\newrefformat{fact}{\savehyperref{#1}{Fact~\ref*{#1}}}
\newrefformat{prop}{\savehyperref{#1}{Proposition~\ref*{#1}}}
\newrefformat{claim}{\savehyperref{#1}{Claim~\ref*{#1}}}
\newrefformat{relax}{\savehyperref{#1}{Relaxation~\ref*{#1}}}
\newrefformat{red}{\savehyperref{#1}{Reduction~\ref*{#1}}}
\newrefformat{part}{\savehyperref{#1}{Part~\ref*{#1}}}
\newrefformat{prob}{\savehyperref{#1}{Problem~\ref*{#1}}}
\newrefformat{ass}{\savehyperref{#1}{Assumption~\ref*{#1}}}

\usepackage{stmaryrd}

\title{Generalizing the Hypergraph Laplacian via
a Diffusion Process with Mediators}

\author{
T-H. Hubert Chan\thanks{Department of Computer Science, the University of Hong Kong. \texttt{hubert@cs.hku.hk, liangzb@connect.hku.hk} }
\and Zhibin Liang\samethanks[1]
}
\date{}

\begin{document}

\begin{titlepage}

\maketitle

\begin{abstract}

In a recent breakthrough STOC~2015 paper, a continuous diffusion process was considered on hypergraphs (which has been refined in a recent JACM 2018 paper) to define a Laplacian operator, whose spectral properties satisfy the celebrated Cheeger's inequality. However, one peculiar aspect of this diffusion process is that each hyperedge directs flow only from vertices with the maximum density to those with the minimum density, while ignoring vertices having strict in-beween densities.

In this work, we consider a generalized diffusion process, in which vertices in a hyperedge can act as mediators to receive flow from vertices with maximum density and deliver flow to those with minimum density.  We show that the resulting Laplacian operator still has a second eigenvalue satsifying the Cheeger's inequality.

Our generalized diffusion model shows that there is a family of operators whose spectral properties are related to hypergraph conductance, and provides a powerful tool to enhance the development of spectral hypergraph theory. Moreover, since every vertex can participate in the new diffusion model at every instant, this can potentially have wider practical applications.

\end{abstract}

\thispagestyle{empty}
\end{titlepage}

\section{Introduction}
\label{sec:intro}

Spectral graph theory, and specifically, the well-known Cheeger's inequality
give a relationship between the edge expansion properties of a graph
and the eigenvalues of some appropriately defined matrix~\cite{alon1986eigenvalues,alon1985lambda1}.
Loosely speaking,
for a given graph, its edge expansion or \emph{conductance} gives a lower bound
on the ratio of the number of edges leaving a subset~$S$ of vertices to the sum of vertex degrees in $S$.
It is natural that graph conductance
is studied in the context of
graph partitioning or clustering~\cite{jacm/KannanVV04,colt/MakarychevMV15,PengSZ15}, whose goal
is to minimize the weight of edges crossing different clusters with respect
to intra-cluster edges.
The reader can refer to the standard references~\cite{chung1997spectral,hoory2006expander} for
an introduction to spectral graph theory.

\noindent \textbf{Recent Generalization to Hypergraphs.}
In an edge-weighted hypergraph $H = (V,E,w)$,
an edge $e \in E$ is a non-empty subset of $V$.
The edges have positive weights indicated by 
$w : E \rightarrow \R_+$.
The weight of each vertex $v \in V$ 
is its weighted degree $w_v := \sum_{e \in E: v \in e} w_e$.
A subset $S$ of vertices has
weight $w(S) := \sum_{v \in S} w_v$, and the edges it cuts is $\partial S := \{e \in E : $ $e$
intersects both $S$ and $V \setminus S \}$.

The conductance of $S \subseteq V$
is defined as $\phi(S) := \frac{w(\partial S)}{w(S)}$.  The conductance
of $H$ is defined as:

\begin{equation}
\phi_H := \min_{\eset \subsetneq S \subsetneq V} \max\{\phi(S), \phi(V \setminus S)\}.
\label{eq:hyper_exp}
\end{equation}

Until recently, it was an open problem to define a spectral model for hypergraphs.
In a breakthrough STOC~2015 paper, Louis~\cite{louis2015hypergraph} considered a 
continuous diffusion process
on hypergraphs (which has been refined in a recent JACM paper~\cite{chan2018jacm}),
and defined an operator 
$\Lo_w f := - \frac{df}{dt}$,
where $f \in \R^V$ is some appropriate vector associated with the diffusion process.
As in classical spectral graph theory,
$\Lo_w$ has non-negative eigenvalues, and the all-ones vector $\one$
is an eigenvector with eigenvalue 0.  Moreover, the operator $\Lo_w$
has a second eigenvalue~$\gamma_2$, and the Cheeger's inequality can be recovered\footnote{
In fact, as shown in this work, a stronger upper bound holds:
$\phi_H \leq  \sqrt{2 \gamma_2}$.} for hypergraphs:

\centerline{$
    \frac{\gamma_2}{2} \leq \phi_H \leq 2 \sqrt{\gamma_2}.
$}

\noindent \textbf{Limitation of the Existing Diffusion Model~\cite{chan2018jacm,louis2015hypergraph}.}
Suppose at some instant, each vertex has some \emph{measure} that is given by
a measure vector $\vp \in \R^V$.  A corresponding \emph{density} vector $f \in \R^V$
is defined by $f_u := \frac{\vp_u}{w_u}$, for each $u \in V$.  Then, at this instant,
each edge $e \in E$ will cause measure to flow
from vertices $S_e(f) := \argmax_{s \in e} f_s$
having the maximum density
to vertices $I_e(f) := \argmin_{i \in e} f_i$ having
the minimum density, at a rate of $w_e \cdot \max_{s,i \in e} (f_s - f_i)$.
Observe that there can be more than one vertex achieving the maximum or the minimum
density in an edge, and a vertex can be involved with multiple number of edges.
As shown in~\cite{chan2018jacm}, it is non-trivial to determine
the net rate of incoming measure for each vertex.

One peculiar aspect of this diffusion process is that each edge~$e$
only concerns its vertices having the maximum or the minimum density,
and ignores the vertices having strict in-between densities.
Even though this diffusion process leads to a theoretical treatment
of spectral hypergraph properties, its practical use is somehow limited,
because it would be considered more natural if vertices having intermediate
densities in an edge also take part in the diffusion process.

For instance, in a recent work on semi-supervised learning
on hypergraphs~\cite{ZhangHTC17}, the diffusion operator is used to 
construct an update vector that changes only the solution values of
vertices attaining the maximum or the minimum in hyperedges.
Therefore, we consider the following open problem in this work:

\emph{Is there a diffusion process on hypergraphs that involves all vertices in every edge at every instant
such that the resulting operator still retains desirable spectral properties?}

\subsection{Our Contribution and Results.}

\noindent \textbf{Generalized Diffusion Process with Mediators.} We consider
a diffusion process where for each edge~$e$,
a vertex~$j \in e$ can act as a \emph{mediator} that 
receives flow from vertices in $S_e(f)$ and delivers flow to~$I_e(f)$.
Formally, we denote $[e] := e \cup \{0\}$,
where $0$ is a special index that does not refer to any vertex.
Each edge~$e$ is equipped with non-negative constants
$(\beta^e_j : j \in [e])$ such that $\sum_{j \in [e]} \beta^e_j = 1$.
Intuitively, for $j =0$,  $\beta^e_0$ refers to the
effect of flow going directly from $S_e(f)$ to $I_e(f)$;
for each vertex~$j \in e$, $\beta^e_j$ refers to the significance
of~$j$ as a mediator between $S_e(f)$ and $I_e(f)$.  The complete
description of the diffusion rules is in Definition~\ref{defn:rules}.
Here are some interesting special cases captured by the new diffusion model.

\begin{compactitem}
\item For each $e \in E$, $\beta^e_0 = 1$.  This is the
existing model in~\cite{chan2018jacm,louis2015hypergraph}.

\item For each $e \in E$, there is some $j_e \in e$ such that $\beta^e_{j_e} = 1$,
i.e., each edge has one special vertex that acts as its mediator
who regulates all flow within the edge.

\item For each $e \in E$, for each $j \in e$, $\beta^e_j = \frac{1}{|e|}$, i.e.,
every vertex in an edge are equally important as mediators.
\end{compactitem}

\begin{theorem}[Recovering Cheeger's Inequality via Diffusion Process with Mediators]
\label{th:main}
Given a hypergraph $H = (V, E, w)$
and mediator constants $(\beta^e_j: e \in E, j \in [e])$,
the diffusion process in Definition~\ref{defn:rules}
defines an operator $\Lo_w f := - \frac{df}{dt}$
that has a second eigenvalue~$\gamma_2$ satisfying
$
\frac{\gamma_2}{2} \leq \phi_H \leq 2 \sqrt{ \gamma_2}
$, where $\phi_H$ is the hypergraph conductance defined in~(\ref{eq:hyper_exp}).
\end{theorem}

\noindent \textbf{Impacts of New Diffusion Model.}  Our generalized
diffusion model shows that there is a family of operators whose spectral properties
are related to hypergraph conductance.  On the theoretical aspect,
this provides a powerful tool to enhance the development of
spectral hypergraph theory.

On the practical aspect, as mentioned earlier,
in the context of semi-supervised learning~\cite{hein2013total,ZhangHTC17},
the following minimization convex program is considered:
the objective function  is $\Q(f) := \langle f, \Lo_w f \rangle_w$,
and the $f$ values of labeled vertices are fixed.
For an iterative method to solve the convex program,
our new diffusion model can possibly lead to an update vector
that modifies every coordinate in the current solution, thereby potentially improving
the performance of the solver.

\subsection{Related Work}

\noindent \emph{Other Works on Diffusion Process and Spectral Graph Theory.}
Apart from the most related aforementioned works~\cite{chan2018jacm,louis2015hypergraph}
that we have already mentioned,
similar diffusion models (without mediators) have been considered for directed normal graphs~\cite{yoshida2016nonlinear}
and directed hypergraphs~\cite{CTWZ2017} to define operators whose spectral properties
are analyzed.

\medskip
\noindent \emph{Higher-Order Cheeger Inequalities.}
For normal graphs,
Cheeger-like inequalities 
to relate higher-order spectral properties
with multi-way edge expansion have been investigated~\cite{KwokLL16,kwok2013improved,lee2014multiway,louis2014approximation,louis2011algorithmic,louis2012many}.
On the other hand, for hypergraphs,
the higher-order spectral properties of the diffusion operator are still unknown.
However, Cheeger-like inequalities
have been derived in terms of the discrepancy ratio~\cite{chan2018jacm},
but not related to the spectral properties of the diffusion operator.

\section{Preliminaries}
\label{sec:hyper-notation}

We consider an edge-weighted hypergraph $H = (V,E,w)$.
Without loss of generality, we assume that 
the weight $w_i := \sum_{e \in E: i \in e} w_e$
of each vertex~$i \in V$ is positive,
since any vertex with zero weight can be removed.
We use $\W \in \R^{V \times V}$ to denote the diagonal matrix whose
$(i,i)$-th entry is the vertex weight~$w_i$;
we let $\I$ denote the identity matrix.

We use $\R^V$ to denote the set of column vectors.
Given $f \in \R^V$,
we use $f_u$ or $f(u)$ 
to indicate the coordinate
corresponding to $u \in V$.
We use $A^\T$ to denote the transpose of a matrix $A$.

We use $\one \in \R^V$ to denote the vector having $1$ in every coordinate. 
For a vector $x \in \R^V$, we define its support as the set of coordinates at which $x$ is non-zero, i.e.
$\supp(x) \defeq \set{i : x_i \neq 0}$.

We use $\chi_S \in \{0,1\}^V$ to denote the indicator vector of the set $S \subset V$, i.e.,
$\chi_S(v) = 1 $ \emph{iff} $v \in S$.

Recall that the conductance $\phi_H$ of a hypergraph $H$ 
is defined in (\ref{eq:hyper_exp}).
We drop the subscript whenever the
hypergraph is clear from the context.

\noindent \textbf{Generalized Quadratic Form.} For each edge $e \in E$,
we denote $[e] := e \cup \{0\}$,
where $0$ is a special index that does not correspond to any vertex.
Then, each edge~$e$ is associated with non-negative constants $(\beta^e_j: j \in [e])$
such that $\sum_{j \in [e]} \beta^e_j = 1$.
The \emph{generalized quadratic form} is defined for each $f \in \R^V$ as:

\begin{equation}
\textstyle \Q(f) := \sum_{e\in E}\w{e}\{
\beta^e_0\max_{s,i\in e}\left(\f{s}-\f{i}\right)^2+\sum_{j\in e}\beta^e_{j}[(\max_{s\in e}\f{s}-\f{j})^2
+(\f{j}-\min_{i\in e}\f{i})^2 ]
\}.
\label{eq:quad}
\end{equation}

\noindent For each non-zero $f \in \R^V$, its \emph{discrepancy ratio} is defined as $\D_w(f) := \frac{\Q(f)}{\sum_{u \in V} \w{u} \f{u}^2}$.

\noindent \emph{Remark.}  Observe that for each $S \subseteq V$, the corresponding
indicator vector $\chi(S) \in \{0,1\}^V$ satisfies $\Q(\chi(S)) = w(\partial S)$.
Hence, we have $\D_w(\chi(S)) = \phi(S)$.

\noindent \textbf{Special Case.} We denote
$\Q^0(f) := \sum_{e \in E} w_e \max_{s,i\in e}\left(\f{s}-\f{i}\right)^2$
for the case when $\beta^e_0 =1$ for all $e$, which was considered in~\cite{chan2018jacm}.
As we shall see later, for $j \in e$,
the weight $\beta^e_j$ denotes the significance
of vertex~$j$ as a ``mediator'' in the diffusion process to direct measure from vertices
of maximum density to those with minimum density.
As in~\cite{chan2018jacm}, we consider three isomorphic spaces as follows.

\noindent \textbf{Density Space.} This is the space
associated with the quadratic form $\Q$.
For $f,g \in \R^V$, the inner product
is defined as $\langle f, g \rangle_w := f^\T \W g$,
and the associated norm is $\| f \|_w := \sqrt{\langle f, f \rangle_w}$.  We use $f \perp_w g$ to denote $\langle f, g \rangle_w = 0$.

\noindent \textbf{Normalized Space.} 
Given $f \in \R^V$ in the density space,
the corresponding vector in the normalized space is $x := \Wh f$.
The normalized discrepancy ratio is
$\Dc(x) := \D_w(\Wmh x) = \D_w(f)$.

In the normalized space, the usual $\ell_2$ inner product and norm are used.
Observe that if $x$ and $y$ are the corresponding
normalized vectors for $f$ and $g$ in the density space,
then $\langle x, y \rangle = \langle f, g \rangle_w$.

\noindent \emph{Towards Cheeger's Inquality.}
Using the inequality $a^2 + b^2 \leq (a+b)^2 \leq 2(a^2 + b^2)$
for non-negative $a$ and $b$,
we conclude  that $\Q(f) \leq \Q^0(f) \leq 2 \Q(f)$ for all $f \in \R^V$.
This immediately gives a partial result of Theorem~\ref{th:main}.

\begin{lemma}[Cheeger's Inequality for Quadratic Form]
\label{lemma:cheeger}
Suppose $\gamma_2 := \min_{\zero \neq f \perp_w \one} \frac{\Q(f)}{\|f\|_w^2}$. Then, we have
$\frac{\gamma_2}{2} \leq \phi_H \leq 2 \sqrt{\gamma_2}$,
where $\phi_H$ is the hypergraph conductance defined in (\ref{eq:hyper_exp}).
\end{lemma}

\begin{proof}
Denote $\gamma^0_2 := \min_{\zero \neq f \perp_w \one} \frac{\Q^0(f)}{\|f\|_w^2}$.
Then, the result from~\cite{chan2018jacm} and an improved upper bound
in Appendix~\ref{sec:up_bound_cheeger} gives: $\frac{\gamma^0_2}{2} \leq \phi_H \leq \sqrt{2 \gamma^0_2}$.
Finally, $\Q \leq \Q^0 \leq 2 \Q$ implies that $\gamma_2 \leq \gamma^0_2 \leq 2 \gamma_2$.
Hence, the result follows.
\qed
\end{proof}

\noindent \textbf{Goal of This Paper.}
In view of Lemma~\ref{lemma:cheeger},
the most technical part of the paper is to
define an operator\footnote{In
the literature, the weighted Laplacian is actually $\W \Lo_w$ in our notation.  Hence,
to avoid confusion, we restrict the term Laplacian to the normalized space.} 
$\Lo_w: \R^V \ra \R^V$ such that
$\langle f, \Lo_w f \rangle_w = \Q(f)$, and show that 
$\gamma_2$ defined in Lemma~\ref{lemma:cheeger} is indeed an eigenvalue of $\Lo_w$.
To achieve this,
we shall consider a diffusion process in the following measure space.

\noindent \textbf{Measure Space.}  
Given a density vector $f \in \R^V$,
multiplying each coordinate with its corresponding weight
gives the measure vector $\vp := \W f$.  
Observe that a vector in the measure space can have negative coordinates.
We do not consider inner product explicitly in this space, and 
so there is no special notation for it.

\noindent \textbf{Transformation between Different Spaces.}
We use the Roman letter $f$ for vectors in
the density space, $x$ for vectors in the
normalized space, and Greek letter $\vp$
for vectors in the measure space.
Observe that 
an operator defined on one space
induces operators on
the other two spaces.
For instance, if $\Lo$ is an operator defined on the measure space,
then $\Lo_w := \Wm \Lo \W$  is the corresponding
operator on the density space 
and $\Lc := \Wmh \Lo \Wh$ is the one on the normalized space.
Moreover, all three operators have the same eigenvalues.
Recall that the Rayleigh quotients are defined
as $\ray_w(f) := \frac{\langle f, \Lo_w f \rangle_w}{\langle f, f \rangle_w}$
and $\rayc(x) := \frac{\langle x, \Lc x \rangle}{\langle x, x \rangle}$. For $\Wh f = x$, we have $\ray_w(f) = \rayc(x)$.

\section{Diffusion Process with Mediators}
\label{sec:diffusion}

\noindent \textbf{Intuition.}  Given an
edge-weighted hypergraph $H=(V,E,w)$,
suppose at some instant, each vertex has some measure given by the
vector $\vp \in \R^V$, whose corresponding
density vector is $f = \Wm \vp$.  The idea of a diffusion process
is that within each edge~$e \in E$, measure should flow
from vertices with higher densities to those with lower densities,
and the rate of flow has a positive correlation with
the difference in densities and the strength of the edge~$e$ given by $w_e$.
If the diffusion process is well-defined, then an operator
on the density space can be defined as $\Lo_w f := - \frac{df}{dt}$.
This induces the Laplacian operator $\Lc := \Wh \Lo_w \Wmh$
on the normalized space.

In previous work~\cite{chan2018jacm},
within an edge, measure only
flows from vertices $S_e(f) := \argmax_{s \in e} f_s \subseteq e$
having the maximum density to those
$I_e(f) := \argmin_{i \in e} f_i$ having minimum densities,
where the rate of flow is $w_e \cdot \max_{s,i \in e} (f_s - f_i)$.
If all $f_u$'s for an edge $e$ are equal,
then we use the convention that $I_e(f) = S_e(f) = e$.
Note that vertices $j \in e \setminus (S_e(f) \cup I_e(f))$
with strict in-between densities
do not participate due to edge~$e$ at this instant.

\noindent \textbf{Generalized Diffusion Process with
Mediators.} In some applications as mentioned
in Section~\ref{sec:intro}, it might be more natural
if every vertex in an edge~$e$ plays some role
in diverting flow from $S_e(f)$ to $I_e(f)$.  In our new
diffusion model, each edge~$e$ is associated
with constants $(\beta^e_j: j \in [e])$ such that
$\sum_{j \in [e]} \beta^e_j = 1$.

Here, $0$ is a special index and the parameter $\beta^e_0$
corresponds to the significance of measure flowing directly
from $S_e(f)$ to $I_e(f)$.
For $j \in e$, $\beta^e_j$ indicates the significance
of vertex~$j$ as a ``mediator'' to receive measure
from $S_e(f)$ and deliver measure to $I_e(f)$.
The formal rules are given as follows.

\begin{definition}[Rules of Diffusion Process]
\label{defn:rules}
Suppose at some instant the system is in a state
given by the density vector $f \in \R^V$,
with measure vector $\vp =\W f$.  
Then, at this instant,
measure is transferred between vertices according to the following rules.
For $u \in e$ and
$j \in [e]$,	the pair $(e,j)$ imposes some rules on the diffusion process;
let $\vp'_u(e,j)$ be the net rate of measure flowing into vertex~$u$ due to 
the pair~$(e,j)$.

\begin{compactitem}

\item [\textsf{R(0)}] For each vertex $u \in V$, the density changes according to
the net rate of incoming measure divided by its weight:

		\centerline{$
		w_u \frac{d f_u}{d t} = \vp'_u := \sum_{e \in E: u \in e} \sum_{j \in [e]} \vp'_u(e,j).
		$}

\item [\textsf{R(1)}] 

We have $\vp'_u(e,j) < 0$ and $u \neq j$ implies that $u \in S_e(f)$.

Similarly, $\vp'_u(e,j) > 0$ and $u \neq j$ implies that $u \in I_e(f)$.

		\item [\textsf{R(2)}] Each edge $e \in E$ and $j \in [e]$,
		the rates of flow satisfy the following.
		
		For $j = 0$, the rate of flow from $S_e(f)$ to $I_e(f)$ due to $(e,0)$ is:
		
		\centerline{$
		- \sum_{u \in S_e(f)} \vp'_u(e,0)  =  w_e \cdot \beta^e_0 \cdot \max_{s,i \in e}(f_{s}-f_{i})
		= \sum_{u \in I_e(f)} \vp'_u(e,0).
		$}

		For $j \in e$, the rate of flow from $S_e(f)$ to $j$ due to $(e,j)$ is:

		\centerline{$
        - \sum_{u \in S_e(f)} \vp'_u(e,j)  = w_e \cdot \beta^e_j  \cdot (\max_{s \in e} f_{s} -f_{j});
        $}

		the rate of flow from $j$ to $I_e(f)$ due to $(e,j)$ is:
		
		\centerline{$
		\sum_{u \in I_e(f)} \vp'_u(e,j)  = w_e \cdot \beta^e_j  \cdot (f_{j}- \min_{i \in e}f_{i}).
		$}
		
		Then the net rate of flow received by $j$ due to $(e,j)$ is:
		
		\centerline{$
		w_e \cdot \beta^e_j  \cdot (\max_{s \in e} f_s + \min_{i \in e} f_i - 2 f_j)
		=  \vp'_j(e,j).
		$}
		
	\end{compactitem}
\end{definition}

\noindent \textbf{Existence of Diffusion Process.} The diffusion rules
in Definition~\ref{defn:rules} are much more complicated than
those in~\cite{chan2018jacm}.  It is not immediately obvious whether
such a process is well-defined.  However, the techniques in~\cite{CTWZ2017}
can be employed.  Intuitively, by repeatedly applying the procedure
described in Section~\ref{sec:disp}, all higher-order derivatives of the density vector
can be determined, which induce an equivalence relation on~$V$
such that vertices in the same equivalence class will have the same density in infinitesimal time.
This means the hypergraph can be reduced to a simple graph, in which the diffusion process is
known to be well-defined.  However, to argue this formally is non-trivial,
and the reader can refer to the details in~\cite{CTWZ2017}.

As in~\cite{chan2018jacm},
if we define an operator using the diffusion process in Definition~\ref{defn:rules},
then the resulting Rayleigh quotient coincides with the discrepancy ratio.
The proof of the following lemma is deferred to Appendix~\ref{sec:ray_disc}.

\begin{lemma}
[Rayleigh Quotient Coincides with Discrepancy Ratio]
\label{lemma:ray_disc}
Suppose $\Lo_w$ on the density space is defined as
$\Lo_w f := - \frac{df}{dt}$ by
the rules in Definition~\ref{defn:rules}.
Then, 
the Rayleigh quotient associated with $\Lo_w$ satisfies
that for any $f$ in the density space,
$\ray_w(f) = \D_w(f)$.
By considering the isomorphic normalized space,
we have for each $x$, $\rayc(x) = \Dc(x)$.
\end{lemma}

\section{Computing the First Order Derivative in the Diffusion Process}
\label{sec:disp}

In Section~\ref{sec:diffusion}, we define a diffusion process,
whose purpose is to define an operator $\Lo_w f := - \frac{df}{dt}$,
where $f \in \R^V$ is in the density space.  In this section,
we show that the diffusion rules uniquely determine the first order
derivative vector $\frac{df}{dt}$; moreover, we give an algorithm to compute it.

\noindent \textbf{Infinitesimal Considerations.}
In Definition~\ref{defn:rules}, 
if a vertex $u$ is losing measure due to the pair $(e, j)$
and $u \neq j$, then $u$ must be in $S_e(f)$.
However, $u$ must also continue to stay in $S_e(f)$ in infinitesimal time;
otherwise, if $u$ is about to leave $S_e(f)$, then 
$u$ should no longer lose measure due to $(e,j)$.
Hence, the vertex~$u$ should have the maximum first-order derivative of $f_u$
among vertices in $S_e(f)$.
A similar rule should hold when $u$ is gaining measure due to $(e,j)$ and $u \neq j$.
This is formalized as the first-order variant of \textsf{(R1)}:

Rule \textsf{(R3)} First-Order Derivative Constraints:

If $\vp'_u(e,j)<0$ and $u \neq j$, then $u \in \argmax_{s \in S_e(f)} \frac{d f_s}{dt}$.

If $\vp'_u(e, j)>0$ and $u \neq j$, then $u \in \argmin_{i \in I_e(f)} \frac{d f_i}{dt}$.

\noindent \textbf{Considering Each Equivalence Class $U$ Independently.}
As in~\cite{chan2018jacm},
we consider the equivalence relation induced by $f \in \R^V$,
where two vertices $u$ and $v$ are in the same equivalence class \emph{iff} $f_u = f_v$.
For vertices in some equivalence class~$U$,
their current $f$ values are the same, but their values could be about
to be separated because their first derivatives might be different.

\noindent \textbf{Subset with the Largest First Derivative: Densest Subset.}
Suppose $X \subseteq U$ are the vertices having the largest derivative in $U$.
Then, these vertices should receive or contribute rates of measure in each of the following cases.

\begin{compactitem}

\item[1.] The subset $X$ receives measure due to edges $I_X := \{e \in E: I_e(f) \subseteq X\}$,
because the corresponding vertices in $X$ continue to have minimum $f$ values in these edges;
we let $c^I_e \geq 0$ be the rate of measure received by $I_e(f)$ due to $(e,j)$
for $j \notin I_e(f)$.

\item[2.] The subset $X$ contributes measure due to edges $S_X := \{e \in E: S_e(f) \cap X \neq \emptyset\}$,
because the corresponding vertices in $X$ continue to have maximum $f$ values in these edges;
we let $c^S_e \geq 0$ be the rate of measure delivered by $S_e(f)$ due to $(e,j)$
for $j \notin S_e(f)$.

\item[3.] Each $j \in X$ receives or contributes measure due to all $(e,j)$'s such that $e \in E$ and $j \in e$;
we let $c_j \in \R$ be the net rate of measure received by vertex~$j$
due to $(e,j)$ for all $e \in E$ such that $j \in e$.

\end{compactitem}

Hence, the net rate of measure received by $X$
is

\centerline{
$\mathfrak{C}(X) := \sum_{e \in I_X} c^I_e - \sum_{e \in S_X} c^S_e + \sum_{j \in X} c_j$.
}

Therefore, given an instance $(U, I_U, S_U)$,
the problem is to find a maximal subset~$P\subseteq U$ with
the largest density $\delta(P) := \frac{\mathfrak{C}(P)}{w(P)}$,
which will be the $\frac{df}{dt}$ values for the vertices in $P$.
For the remaining vertices in $U$,
the sub-instance $(U \setminus P, I_U \setminus I_P, S_U \setminus S_P)$
is solved recursively.  The procedure and the precise parameters are given in Fig.~\ref{fig:define_r}.
Efficient algorithms for this densest subset problem are described in~\cite{DanischCS17,chan2018jacm}.

\begin{figure}[H]
\begin{tabularx}{\columnwidth}{|X|}
\hline

\vspace{5pt}
Given a hypergraph $H=(V,E,w)$ and a vector $f\in\R^V$ in the density space,
define an equivalence relation on $V$ such that $u$ and $v$ are in the same equivalence class \emph{iff} $f_u=f_v$.
We consider each such equivalence class $U \subseteq V$ and define the $r = \frac{df}{dt}$ values for vertices in $U$
as follows.

\begin{enumerate}
	
	\item 
Denote $E_U:=\{e\in E: U \cap [I_e(f) \cup S_e(f)]\not= \emptyset  \}$.

For $e \in E$, define

$c^I_e := w_e \cdot [\beta^e_0  \cdot \max_{s,i \in e} (f_s - f_i)
+ \sum_{j \in e} \beta^e_j \cdot (f_j - \min_{i \in e} f_i)],$

$c^S_e := w_e \cdot [\beta^e_0  \cdot \max_{s,i \in e} (f_s - f_i)
+ \sum_{j \in e} \beta^e_j \cdot (\max_{s \in e} f_s - f_j)];$

for $j \in V$, define
$c_j := \sum_{e \in E: j \in e} \beta^e_j \cdot 
w_e \cdot (\max_{s \in e} f_s + \min_{i \in e} f_i - 2 f_j).$

For $X\subseteq U$, define $I_X:=\{e\in E_U: I_e(f) \subseteq X\}$, $S_X:=\{e\in E_U:  S_e(f) \cap X\not=\emptyset\}$.

Denote $\mathfrak{C}(X):=\sum_{e\in I_X} c^I_e
-\sum_{e\in S_X}c^S_e
+\sum_{j \in X} c_j$
and $\delta(X):=\frac{\mathfrak{C}(X)}{w(X)}$.

\item Find $P\subseteq U$ such that $\delta(P)$ is maximized.
For all $u\in P$, set $r_u:=\delta(P)$.

\item Recursively, find the $r$ values for the remaining vertices in $U':=U\setminus P$ using $E_{U'}:=E_U\setminus (I_P\cup S_P)$.

\end{enumerate}

\\
\hline 
\end{tabularx}
\caption{Procedure to compute $r=\frac{df}{dt}$}
\label{fig:define_r}
\end{figure}

The next lemma shows that
the procedure in Fig.~\ref{fig:define_r}
returns a vector $r \in \R^V$ 
that coincides with the first-order derivative $\frac{df}{dt}$
of the density vector obeying rules~\textsf{(R0)} to \textsf{(R3)}.
This implies that these rules uniquely determine the first-order derivative.
Given $f \in \R^V$ and $r = \frac{df}{dt}$,
we denote
$r_S(e) := \max_{u \in S_e(f)} r_u$ and $r_I(e) := \min_{u \in I_e(f)} r_u$.

\begin{lemma}[Densest Subset Problem Determines First-Order Deriative]
\label{lemma:define_lap}
Given a density vector $f \in \R^V$,
rules~\textsf{(R0)} to \textsf{(R3)}
uniquely determine $r = \frac{df}{dt} \in \R^V$,
which can be found by the procedure
described in Fig.~\ref{fig:define_r}.
Moreover,
$\sum_{e \in E} c^I_e\cdot r_I(e)
-\sum_{e \in E} c^S_e\cdot r_S(e)
+\sum_{j\in V}c_j\cdot r_j= \sum_{u \in V} \vp'_u r_u = \|r\|^2_w$.
\end{lemma}

\begin{proof}
Using the same approach as in~\cite{chan2018jacm},
we consider each equivalence class $U$ in Fig.~\ref{fig:define_r},
where all vertices in a class have the same $f$ values.

For each such equivalence
class $U \subset V$,
define $I_U := \{e \in E: U\cap I_e(f)\not=\emptyset\}$,
$S_U := \{e \in E: U\cap S_e(f)\not=\emptyset\}$.
Notice that each $e$ can only be in exactly
one of $I_U$ and $S_U$.

\noindent \textbf{Considering Each Equivalence Class $U$.}
Suppose $T$ is the set of vertices within $U$ that have
the maximum first-order derivative $r = \frac{df}{dt}$.
It suffices to show that $T$ is the maximal densest subset
in the densest subset instance $(U, I_U  \cup S_U)$
defined in Fig.~\ref{fig:define_r}.

Because of rule~\textsf{(R3)},
the rate of net measure received by~$T$ 
is $\mathfrak{C}(T)$.  Hence,
all vertices~$u \in T$ have $r_u = \frac{\mathfrak{C}(T)}{w(T)}$.

Next, suppose~$P$ is the maximal densest subset
found in Fig.~\ref{fig:define_r}.
Observe that the net rate of measure 
entering $P$ is at least $\mathfrak{C}(P)$.
Hence, there exists some vertex~$v \in P$
such that $\frac{\mathfrak{C}(P)}{w(P)} \leq r_v \leq \frac{\mathfrak{C}(T)}{w(T)}$,
where the last inequality follows from the definition of $T$.

Since $P$ is the maximal densest subset, it follows that in the above inequality,
actually all equalities hold and all vertices in $P$ have the same $r$ value.  
In general, the maximal densest subset
contains all densest subsets, and it follows that $T \subseteq P$.  
Since all vertices in $P$ have the maximum $r$ value within $U$, we conclude that $P = T$.

\noindent \emph{Recursive Argument.} Hence, it follows that
the set $T$ can be uniquely identified
in Fig.~\ref{fig:define_r}
as the set of vertices having maximum $r$ values,
which is also the unique maximal densest subset.
Then, the argument can be applied
recursively for the smaller instance with
$U' := U \setminus T$, $I_{U'} := I_U \setminus I_T$,
$S_{U'} := S_U \setminus S_T$.

\vspace{10pt}

\noindent \textbf{Claim.}
$
\sum_{e \in E} c^I_e\cdot r_I(e)
-\sum_{e \in E} c^S_e\cdot r_S(e)
+\sum_{j\in V}c_j\cdot r_j
= \sum_{u \in V} \vp'_u r_u = \|r\|^2_w.
$

Consider some $T$ defined above with $\delta := \delta(T) = r_u$, for $u \in T$.

Observe that
\begin{align}
\begin{split}
\textstyle
\sum_{u \in T} \vp'_u r_u
&=\left(
\textstyle
\sum_{e\in I_T} c^I_e
-\sum_{e\in S_T}c^S_e
+\sum_{j \in T} c_j
\right) \cdot \delta
\\
&\textstyle
= \sum_{e \in I_T} c^I_e \cdot\min_{i\in I_e}r_i
- \sum_{e \in S_T} c^S_e \cdot \max_{s\in S_e}r_s
+\sum_{j\in T}c_j\cdot r_j
\end{split}
\nonumber
\end{align}
where the last equality is due to rule~\textsf{(R3)}.

Observe that every $u \in V$ will be in exactly one such $T$, and
every $e \in E$ will be accounted for exactly once in each of $I_T$ and $S_T$, ranging over all $T$'s.  Hence, summing over all $T$'s gives the result.
\qed
\end{proof}

\section{Spectral Properties of Laplacian}
\label{sec:eigen}

A classical result in spectral graph theory
is that for a $2$-graph whose edge weights
are given by the adjacency matrix $A$,
the parameter $\gamma_2 := \min_{\zero \neq x \perp  \Wh \one} \Dc(x)$ is an eigenvalue of the normalized
Laplacian \mbox{$\Lc := \I - \Wmh A \Wmh$}, 
where a corresponding minimizer $x_2$ is an eigenvector
of $\Lc$.
Observe that $\gamma_2$ is also an eigenvalue
on the operator $\Lo_w := \I - \Wm A$ induced on the density space.

In this section, we generalize the result to hypergraphs.
Observe that any result for the normalized space has an equivalent counterpart in the density space, and vice versa.

\begin{theorem}[Eigenvalue of Hypergraph Laplacian]
    \label{th:hyper_lap}
    For a hypergraph with edge weights $w$,
    there exists a normalized Laplacian $\Lc$ such that the
    normalized discrepancy ratio $\Dc(x)$ coincides
    with the corresponding Rayleigh quotient $\rayc(x)$.
    Moreover, 
    the parameter $\gamma_2 := \min_{\zero \neq x \perp  \Wh \one} \Dc(x)$ is an eigenvalue of $\Lc$,
    where any minimizer $x_2$ is a corresponding eigenvector.
\end{theorem}

Before proving Theorem~\ref{th:hyper_lap},
we first consider the spectral properties
of the normalized Laplacian $\Lc$ induced
by the diffusion process defined in Section~\ref{sec:disp}.

\begin{lemma}[First-Order Derivatives]
    \label{lemma:deriv}
    Consider the diffusion process satisfying rules~\textsf{(R0)}
    to~\textsf{(R3)} on 
    the measure space with $\vp \in \R^V$, which
    corresponds to $f = \Wm \vp$ in the density space.
    Suppose 
    $\Lo_w$ is the induced operator on the density space such that
    $\frac{d f}{d t} = - \Lo_w f$.
    Then, we have the following derivatives.
    
    \begin{compactitem}
        \item[1.] $\frac{d \|f\|^2_w}{dt} = - 2 \langle f, \Lo_w f \rangle_w$.
        \item[2.] $\frac{d \langle f, \Lo_w f \rangle_w}{dt} 
        = - 2 \|\Lo_w f \|^2_w$.
        \item[3.] Suppose $\ray_w(f)$ is the Rayleigh quotient
        with respect to the operator $\Lo_w$ on the density space.
        Then, for $f \neq \zero$, $\frac{d \ray_w(f)}{dt} = -\frac{2}{\|f\|^4_w} \cdot
        (\|f\|^2_w \cdot \|\Lo_w f\|^2_w - \langle f , \Lo_w f \rangle^2_w) \leq 0$,
        by the Cauchy-Schwarz inequality
        on the $\langle \cdot , \cdot \rangle_w$ inner product, where equality
        holds \emph{iff} $\Lo_w f \in \spn(f)$.
				
				By considering a transformation to the normalized space, for any $x \neq \zero$,
				$\frac{d \rayc(x)}{dt} \leq 0$, where equality holds \emph{iff} $\Lc x \in \spn(x)$.

    \end{compactitem}
    
\end{lemma}

\begin{proof}
    For the first statement,
    $\frac{d \|f\|_w^2}{d t} = 2 \langle f, \frac{d f}{d t} \rangle_w
    = - 2 \langle f, \Lo_w f \rangle_w$.
    
    For the second statement,
    from the proof of Lemma~\ref{lemma:ray_disc}
    we have
    $$\langle f, \Lo_w f \rangle_w
    =\textstyle
    \sum_{e\in E}\w{e}\{\beta^e_0\max_{s,i\in e}\left(\f{s}-\f{i}\right)^2+\sum_{j\in e}\beta^e_{j}[\left(\max_{s\in e}\f{s}-\f{j}\right)^2
    +\left(\f{j}-\min_{i\in e}\f{i}\right)^2]\}.$$
    
    Hence, by the Envelope Theorem,
    \begin{align}
    \begin{split}
    \textstyle
    \frac{d \langle f, \Lo_w f \rangle_w}{dt}
    =&
    \textstyle
    2 \sum_{e \in E} \w{e}
    \left[\vphantom{\sum_{j\in e}}
    \beta^e_0\max_{s,i\in e}\left(\f{s}-\f{i}\right)
    \left(\max_{s\in S_e}\frac{df_s}{dt}-\min_{i\in I_e}\frac{df_i}{dt}\right)
    \right]
    \\
    &\textstyle
    +\sum_{j\in e}\beta^e_{j}\left(\max_{s\in e}\f{s}-\f{j}\right)
    \left(\max_{s\in S_e}\frac{df_s}{dt}-\frac{df_j}{dt}\right)
    \\
    &\textstyle
    \left.
    +\sum_{j\in e}\beta^e_{j}\left(\f{j}-\min_{i\in e}\f{i}\right)
    \left(\frac{df_j}{dt}-\min_{i\in I_e}\frac{df_i}{dt}\right)
    \right]
    \\
    =&\textstyle
    2 \sum_{e \in E} \w{e}\left\{
    \left[
    \beta^e_0\max_{s,i\in e}\left(\f{s}-\f{i}\right)
    +\sum_{j\in e}\beta^e_{j}\left(\max_{s\in e}\f{s}-\f{j}\right)
    \right]
    \max_{s\in S_e}\frac{df_s}{dt}
    \right.
    \\
    &\textstyle
    -\left[
    \beta^e_0\max_{s,i\in e}\left(\f{s}-\f{i}\right)
    +\sum_{j\in e}\beta^e_{j}\left(\f{j}-\min_{i\in e}\f{i}\right)
    \right]
    \min_{i\in I_e}\frac{df_i}{dt}
    \\
    &\textstyle
    \left.
    +\sum_{j\in e}\beta^e_{j}\left(2\f{j}-\max_{s\in e}\f{s}-\min_{i\in e}\f{i}\right)
    \frac{df_j}{dt}
    \right\}.
    \\
    =&\textstyle
    2 \left(
    \sum_{e \in E}
    c^I_e\cdot\max_{s\in S_e}r_s
    -\sum_{e \in E}
    c^S_e\cdot\max_{i\in I_e}r_i
    -\sum_{j\in V}c_j\cdot r_j
    \right)
    \end{split}
    \nonumber
    \end{align}

    \noindent where $c^I_e,c^S_e,c_j$ are defined in Fig.~\ref{fig:define_r}.
    From Lemma~\ref{lemma:define_lap},
    this equals $- 2 \|r\|^2_w = - 2 \|\Lo_w f\|^2_w$.
    
    Finally, for the third statement, we have
    $\frac{d}{dt} \frac{\langle f, \Lo_w f \rangle_w}{\langle f, f \rangle_w} = \frac{1}{\|f\|^4_w} (\| f \|^2_w \cdot \frac{d \langle f, \Lo_w f \rangle_w}{ dt} -  \langle f, \Lo_w f \rangle_w \cdot \frac{d \|f\|^2_w}{dt})
    = -\frac{2}{\|f\|^4_w} \cdot
    (\|f\|^2_w \cdot \|\Lo_w f\|^2_w - \langle f , \Lo_w f \rangle^2_w)$,
    where the last equality follows from the first two statements.
\qed
\end{proof}

We next prove some properties
of the normalized Laplacian $\Lc$ with respect to orthogonal
projection in the normalized space.

\begin{lemma}[Laplacian and Orthogonal Projection]
    \label{lemma:lap_proj}
    Suppose $\Lc$ is the normalized Laplacian. 
    Moreover, denote $x_1 := \Wh \one$, and
    let $\Pi$ denote the orthogonal projection
    into the subspace that is orthogonal to $x_1$.
    Then, for all $x$, we have the following:
    \begin{compactitem}
        \item[1.] $\Lc(x) \perp x_1$,
        \item[2.] $\langle x, \Lc x \rangle = \langle \Pi x, \Lc \Pi x \rangle$.
        \item[3.] For all real numbers $a$ and $b$,
        $\Lc(a x_1 + b x) = b \Lc(x)$.
    \end{compactitem}
\end{lemma}

\begin{proof}
    For the first statement, observe that since the diffusion process
    is defined on a closed system, the total measure given by $\sum_{u \in V} \vp_u$ does not change.
    Therefore, $0 = \langle \one, \frac{d \vp}{d t} \rangle = \langle \Wh \one, \frac{d x}{d t} \rangle$,
    which implies that $\Lc x = - \frac{d x}{d t} \perp x_1$.
    
    For the second statement,
    observe that from Lemma~\ref{lemma:ray_disc},
    we have
    $
    \langle x, \Lc x \rangle
    =\sum_{e\in E}\w{e}\{
    \beta^e_0\max_{s,i\in e}(\frac{x_s}{\sqrt{w_s}}-\frac{x_i}{\sqrt{w_i}})^2
    +\sum_{j\in e}\beta^e_{j}[
    (\max_{s\in e}\frac{x_s}{\sqrt{w_s}}-\frac{x_j}{\sqrt{w_j}})^2
    +(\frac{x_j}{\sqrt{w_j}}-\min_{i\in e}\frac{x_i}{\sqrt{w_i}})^2
    ]
    \}
    =\langle (x + \alpha x_1), \Lc (x + \alpha x_1) \rangle,
    $
    where the last equality holds for all real numbers $\alpha$.
    Observe that $\Pi x = x + \alpha x_1$, for some suitable real $\alpha$.

    For the third statement, it is more convenient
    to consider transformation into the density space $f = \Wmh x$.
    It suffices to show that $\Lo_w (a \one + b f) = b \Lo_w(f)$.
		
		Observe that in the diffusion process, only pairwise difference in densities among vertices matters.
		Hence, we immediately have $\Lo_w (a \one + b f) = \Lo_w (b f)$.
		
		For $b \geq 0$, observe that all the rates are scaled by the same factor $b$.
		Hence, we have $\Lo_w(b f) = b \Lo_w(f)$.
		
		Finally, if we reverse the sign of every coordinate of $f$,
		then the roles of $S_e(f)$ and $I_e(f)$ are switched.  Moreover,
		the direction of every component of the measure flow is reversed with the same magnitude.
		Hence, $\Lo_w(-f) = - \Lo_w(f)$, and the result follows.
\qed
\end{proof}

\begin{proofof}{Theorem~\ref{th:hyper_lap}}
This follows the same argument as in~\cite{chan2018jacm}.
    Suppose $\Lc$ is the normalized Laplacian
    induced by the diffusion process in Lemma~\ref{lemma:define_lap}.
    Let $\gamma_2 := \min_{\zero \neq x \perp \Wh \one} \rayc(x)$
    be attained by some minimizer $x_2$.
    We use the isomorphism between the three spaces:
    $\Wmh \vp = x = \Wh f$.
    
    The third statement of Lemma~\ref{lemma:deriv}
    can be formulated in terms of the normalized space,
    which states that $\frac{d \rayc(x)}{d t} \leq 0$,
    where equality holds \emph{iff} $\Lc x \in \spn(x)$.
    
    We claim that $\frac{d \rayc(x_2)}{d t} = 0$.
    Otherwise, suppose $\frac{d \rayc(x_2)}{d t} < 0$.
    From Lemma~\ref{lemma:lap_proj},
    we have $\frac{dx}{dt} = - \Lc x \perp \Wh \one$.
    Hence, it follows that at this moment, the current normalized
    vector is at position $x_2$, and is moving
    towards the direction given by
    $x' := \frac{d x}{dt}|_{x=x_2}$ such that
    $x' \perp \Wh \one$, and $\frac{d \rayc(x)}{d t}|_{x=x_2} < 0$.
    Therefore, for sufficiently small $\eps > 0$,
    it follows that $x_2' := x_2 + \eps x'$ is a non-zero vector
    such that  $x_2' \perp \Wh \one$
    and $\rayc(x_2') < \rayc(x_2) = \gamma_2$, contradicting the definition of $x_2$.
    
    Hence, it follows that $\frac{d \rayc(x_2)}{d t} = 0$,
    which implies that $\Lc x_2 \in \spn(x_2)$.
    Since $\gamma_2 = \rayc(x_2) = \frac{\langle x_2, \Lc x_2 \rangle}{\langle x_2,  x_2 \rangle}$,
    it follows that $\Lc x_2 = \gamma_2 x_2$, as required.
\end{proofof}

{
\bibliographystyle{abbrv}
\bibliography{dihyper}

\begin{thebibliography}{10}

\bibitem{alon1986eigenvalues}
N.~Alon.
\newblock Eigenvalues and expanders.
\newblock {\em Combinatorica}, 6(2):83--96, 1986.

\bibitem{alon1985lambda1}
N.~Alon and V.~D. Milman.
\newblock $\lambda$1, isoperimetric inequalities for graphs, and
  superconcentrators.
\newblock {\em Journal of Combinatorial Theory, Series B}, 38(1):73--88, 1985.

\bibitem{chan2018jacm}
T.~H. Chan, A.~Louis, Z.~G. Tang, and C.~Zhang.
\newblock Spectral properties of hypergraph laplacian and approximation
  algorithms.
\newblock {\em J. {ACM}}, 65(3):15:1--15:48, 2018.

\bibitem{CTWZ2017}
T.~H. Chan, Z.~G. Tang, X.~Wu, and C.~Zhang.
\newblock Diffusion operator and spectral analysis for directed hypergraph
  laplacian.
\newblock {\em CoRR}, abs/1711.01560, 2017.

\bibitem{chung1997spectral}
F.~R. Chung.
\newblock {\em Spectral graph theory}, volume~92.
\newblock American Mathematical Soc., 1997.

\bibitem{DanischCS17}
M.~Danisch, T.~H. Chan, and M.~Sozio.
\newblock Large scale density-friendly graph decomposition via convex
  programming.
\newblock In {\em {WWW}}, pages 233--242. {ACM}, 2017.

\bibitem{hein2013total}
M.~Hein, S.~Setzer, L.~Jost, and S.~S. Rangapuram.
\newblock The total variation on hypergraphs - learning on hypergraphs
  revisited.
\newblock In {\em {NIPS}}, pages 2427--2435, 2013.

\bibitem{hoory2006expander}
S.~Hoory, N.~Linial, and A.~Wigderson.
\newblock Expander graphs and their applications.
\newblock {\em Bulletin of the American Mathematical Society}, 43(4):439--561,
  2006.

\bibitem{jacm/KannanVV04}
R.~Kannan, S.~Vempala, and A.~Vetta.
\newblock On clusterings: Good, bad and spectral.
\newblock {\em J. {ACM}}, 51(3):497--515, 2004.

\bibitem{KwokLL16}
T.~C. Kwok, L.~C. Lau, and Y.~T. Lee.
\newblock Improved cheeger's inequality and analysis of local graph
  partitioning using vertex expansion and expansion profile.
\newblock In {\em {SODA}}, pages 1848--1861. {SIAM}, 2016.

\bibitem{kwok2013improved}
T.~C. Kwok, L.~C. Lau, Y.~T. Lee, S.~O. Gharan, and L.~Trevisan.
\newblock Improved cheeger's inequality: analysis of spectral partitioning
  algorithms through higher order spectral gap.
\newblock In {\em {STOC}}, pages 11--20. {ACM}, 2013.

\bibitem{lee2014multiway}
J.~R. Lee, S.~O. Gharan, and L.~Trevisan.
\newblock Multiway spectral partitioning and higher-order cheeger inequalities.
\newblock {\em Journal of the ACM (JACM)}, 61(6):37, 2014.

\bibitem{louis2015hypergraph}
A.~Louis.
\newblock Hypergraph markov operators, eigenvalues and approximation
  algorithms.
\newblock In {\em {STOC}}, pages 713--722. {ACM}, 2015.

\bibitem{louis2014approximation}
A.~Louis and K.~Makarychev.
\newblock Approximation algorithm for sparsest \emph{k}-partitioning.
\newblock In {\em {SODA}}, pages 1244--1255. {SIAM}, 2014.

\bibitem{louis2011algorithmic}
A.~Louis, P.~Raghavendra, P.~Tetali, and S.~Vempala.
\newblock Algorithmic extensions of cheeger's inequality to higher eigenvalues
  and partitions.
\newblock In {\em {APPROX-RANDOM}}, volume 6845 of {\em Lecture Notes in
  Computer Science}, pages 315--326. Springer, 2011.

\bibitem{louis2012many}
A.~Louis, P.~Raghavendra, P.~Tetali, and S.~Vempala.
\newblock Many sparse cuts via higher eigenvalues.
\newblock In {\em {STOC}}, pages 1131--1140. {ACM}, 2012.

\bibitem{colt/MakarychevMV15}
K.~Makarychev, Y.~Makarychev, and A.~Vijayaraghavan.
\newblock Correlation clustering with noisy partial information.
\newblock In {\em {COLT}}, volume~40 of {\em {JMLR} Workshop and Conference
  Proceedings}, pages 1321--1342. JMLR.org, 2015.

\bibitem{PengSZ15}
R.~Peng, H.~Sun, and L.~Zanetti.
\newblock Partitioning well-clustered graphs: Spectral clustering works!
\newblock In {\em {COLT}}, volume~40 of {\em {JMLR} Workshop and Conference
  Proceedings}, pages 1423--1455. JMLR.org, 2015.

\bibitem{yoshida2016nonlinear}
Y.~Yoshida.
\newblock Nonlinear laplacian for digraphs and its applications to network
  analysis.
\newblock In {\em Proceedings of the Ninth ACM International Conference on Web
  Search and Data Mining}, pages 483--492. ACM, 2016.

\bibitem{ZhangHTC17}
C.~Zhang, S.~Hu, Z.~G. Tang, and T.~H. Chan.
\newblock Re-revisiting learning on hypergraphs: Confidence interval and
  subgradient method.
\newblock In {\em {ICML}}, volume~70 of {\em Proceedings of Machine Learning
  Research}, pages 4026--4034. {PMLR}, 2017.

\end{thebibliography}
}

\newpage

\appendix

\section{Rayleigh Quotient Coincides with Discrepancy Ratio}
\label{sec:ray_disc}

To prove Lemma~\ref{lemma:ray_disc},
we first re-interpret the diffusion rules in Definition~\ref{defn:rules}
by considering the interaction between every pair of nodes.
Observe that the rules sometimes say that some measure is flow from
one subset of vertices to another subset.  Hence, at the moment,
the exact pairwise interactions are not specified.  In fact,
we know that in general, the pairwise interactions are not uniquely determined.
Fig.~\ref{fig:diffusion_framework} captures this
non-deterministic nature of the pairwise interactions.

\begin{figure}[H]
\begin{tabularx}{\columnwidth}{|X|}
\hline

\vspace{5pt}
Given a hypergraph $H(V,E,w)$
and a density vector~$f \in \R^V$,
we analyze the pairwise interaction between vertices
according to Definition~\ref{defn:rules}.

\begin{enumerate}
    \item  We shall describe constraints on the pairwise interaction
		by a symmetric matrix $A_f \in \R^{V \times V}$
		such that for $u, v \in V$,
		the $(u,v)$-th entry $a_{uv}$ means that
		between $u$ and $v$, there is measure flowing
		from the vertex of higher density to that of lower density, 
		at the rate of $a_{uv} \cdot |f_u - f_v|$.
		
		Furthermore, we decompose $a_{uv} := \sum_{e \in E} \sum_{j \in [e]} a^{(e,j)}_{uv}$,
		where the role of the pair $(e,j)$ is described in Definition~\ref{defn:rules}.
		We use the convention that $a^{(e,j)}_{uv} = a^{(e,j)}_{vu}$.

    \item \emph{Pairwise Interaction due to $(e,j)$.} For each $e\in E$
		and $j \in [e]$, we describe the pairwise interaction due to $(e,j)$.

    For $j=0$, by considering the rate of flow from $S_e(f)$ to $I_e(f)$ due to $(e,0)$,
		we infer that the partial weight
		$w_e \beta^e_0$ is somehow distributed among $(s,i)\in S_e(f)\times I_e(f)$.
		In other words, we have the constraint:
		
		$\sum_{(s,i)\in S_e(f)\times I_e(f)} a_{si}^{(e,0)}= w_e \beta^e_0$.

    For $j\in e$, by considering the rate of flow from $S_e(f)$ to $j$ due to $(e,j)$,
		we have the constraint:
		$\sum_{s\in S_e(f)}a_{sj}^{(e,j)}= w_e \beta^e_{j}$.
		
		Similarly, by considering the rate of flow from $j$ to $I_e(f)$ due to $(e,j)$, 
		we have the constraint:
		
		$\sum_{i\in I_e(f)}a_{ji}^{(e,j)}= w_e \beta^e_{j}$.

    For the remaining vertex pairs $(u,v)$ that are not involved with the pair $(e,j)$, we assign $a_{uv}^{(e,j)}=0$.
    
    \item In conclusion, the interaction matrix $A_f$ satisfies the following.
    For $u\not=v$, $A_f(u,v) = \sum_{e \in E} \sum_{j \in [e]} a^{(e,j)}_{uv}$;
		moreover, for each $u \in V$, the row of $A_f$ corresponding to $u$ sums to $w_u$.
		
		\end{enumerate}

Then, the diffusion process is described by $\W \frac{df}{dt} = \frac{d \vp}{dt} = (A_f \Wm - \I) \vp$.
Therefore, the resulting Laplacian operators satisfies
$\Lo(\vp)=(\I-A_f\Wm)\vp$ for the measure space
and $\Lo_w (f) = (\I - \Wm A_f) f$ for the density space.

\\
\hline 
\end{tabularx}
\caption{Constraints on Pairwise Interactions}
\label{fig:diffusion_framework}
\end{figure}

\begin{proofof}{Lemma~\ref{lemma:ray_disc}}
It suffices to show that
$$
\langle f, \Lo_w f \rangle_w =
\sum_{e\in E}\w{e}\left\{\beta^e_0\max_{s,i\in e}\left(\f{s}-\f{i}\right)^2+\sum_{j\in e}\beta^e_{j}\left[\left(\max_{s\in e}\f{s}-\f{j}\right)^2
+\left(\f{j}-\min_{i\in e}\f{i}\right)^2\right]\right\}.
$$

Recall that $\vp = \W f$,
and $\Lo_w = \I - \Wm A_f$,
where $A_f$ satisfies
the constaints in Fig.~\ref{fig:diffusion_framework}.

Hence,
it follows that
\begin{align}
\begin{split}
\langle f, \Lo_w f \rangle_w&=f^\T (\W - A_f) f
= \sum_{uv \in {V \choose 2}} a_{uv} (f_u - f_v)^2
=\sum_{uv \in {V \choose 2}} \sum_{e \in E} \sum_{j \in [e]} a^{(e,j)}_{uv} (f_u - f_v)^2
\\
&=\sum_{uv \in {V \choose 2}} \sum_{e \in E} \left[a^{(e,0)}_{uv} (f_u - f_v)^2+\sum_{j \in e} a^{(e,j)}_{uv} (f_u - f_v)^2\right]
\end{split}
\nonumber
\end{align}

For the first term, we have
\begin{align}
\begin{split}
&\sum_{uv \in {V \choose 2}} \sum_{e \in E} a^{(e,0)}_{uv} (f_u - f_v)^2\\
=&\sum_{uv \in {V \choose 2}} 
\sum_{e \in E:\{uv,vu\}\cap S_e \times I_e\not= \emptyset} a^{(e,0)}_{uv} (f_u - f_v)^2
\\
=&\sum_{e \in E}
\sum_{si \in {e \choose 2}:\{si,is\}\cap S_e \times I_e\not= \emptyset} a^{(e,0)}_{si} (f_s - f_i)^2
\\
=&\sum_{e \in E}
\sum_{si \in S_e \times I_e} a^{(e,0)}_{si} (f_s - f_i)^2
\\
=&\sum_{e \in E} 
w_e\beta^e_0 \max_{s,i\in e}\left(f_{s}-f_{i}\right)^2
\end{split}
\nonumber
\end{align}

For the second term, we have
\begin{align}
\begin{split}
&\sum_{uv \in {V \choose 2}} \sum_{e \in E} \sum_{j \in e} a^{(e,j)}_{uv} (f_u - f_v)^2\\
=&\sum_{uv \in {V \choose 2}} \sum_{e \in E} \left[
\sum_{j\in e: \{uv,vu\}\cap S_e \times j\not= \emptyset} a^{(e,j)}_{uv} (f_u - f_v)^2
+\sum_{j\in e: \{uv,vu\}\cap j \times I_e\not= \emptyset} a^{(e,j)}_{uv} (f_u - f_v)^2
\right]
\\
=&\sum_{e \in E} \sum_{si \in {e \choose 2}} \left[
\sum_{j\in e: \{si,is\}\cap S_e \times j\not= \emptyset} a^{(e,j)}_{si} (f_s - f_i)^2
+\sum_{j\in e: \{si,is\}\cap j \times I_e\not= \emptyset} a^{(e,j)}_{si} (f_s - f_i)^2
\right]
\\
=&\sum_{e \in E}\sum_{j\in e}\left[
\sum_{s\in S_e}a_{sj}^{(e,j)} (f_{s}-f_{j})^2
+\sum_{i\in I_e}a_{ji}^{(e,j)} (f_{j}-f_{i})^2\right]
\\
=&\sum_{e \in E}\sum_{j\in e}\left[
w_e\beta^e_{j} \left(\max_{s\in e}f_{s}-f_{j}\right)^2
+w_e\beta^e_{j} \left(f_{j}-\min_{i\in e}f_{i}\right)^2
\right]
\\
=&\sum_{e \in E}w_e\sum_{j\in e}
\beta^e_{j}\left[\left(\max_{s\in e}f_{s}-f_{j}\right)^2
+ \left(f_{j}-\min_{i\in e}f_{i}\right)^2
\right]
\end{split}
\nonumber
\end{align}

Thus, we conclude
\begin{align}
\begin{split}
\langle f, \Lo_w f \rangle_w
&=\sum_{uv \in {V \choose 2}} \sum_{e \in E} \left[a^{(e,0)}_{uv} (f_u - f_v)^2+\sum_{j \in e} a^{(e,j)}_{uv} (f_u - f_v)^2\right]
\\
&=\sum_{e\in E}\w{e}\left\{\beta^e_0\max_{s,i\in e}\left(\f{s}-\f{i}\right)^2
+\sum_{j\in e}\beta^e_{j}\left[\left(\max_{s\in e}\f{s}-\f{j}\right)^2
+(\f{j}-\min_{i\in e}\f{i})^2\right]\right\}
\end{split}
\nonumber
\end{align}
as required.
\end{proofof}

\section{Improved Upper Bound for Hypergraph Cheeger Inequality}
\label{sec:up_bound_cheeger}

Recall that we denote $\Q^0(f):=\sum_{e\in E}\w{e}
\max_{s,i\in e}\left(\f{s}-\f{i}\right)^2$ for the special case in (\refeq{eq:quad}) when
$\beta_0^e=1$ for all $e\in E$, and $\gamma^0_2 := \min_{\zero \neq f \perp_w \one} \frac{\Q^0(f)}{\|f\|_w^2}$.

\begin{theorem}[Upper Bound for Hypergraph Cheeger Inequalities]
    \label{thm:hyper-cheeger}
    Given an edge-weighted hypergraph $H$, we have:
    \[ \phi_H \leq 
    \min_{\zero \not=f \perp_w \one} \sqrt{\frac{2\Q^0(f)}{\|f\|_w^2}}
    =\sqrt{2\gamma^0_2},\]
    where $\phi_H$ is the hypergraph conductance defined in~(\ref{eq:hyper_exp}).
\end{theorem}

\begin{proposition}
    \label{prop:find_hyper_cut}
    Given an edge-weighted hypergraph $H = (V,E,w)$
    and a non-zero vector $f \in \R^{V}$ such that $f \perp_w \one$,
    there exists a set $S\subseteq \supp(f)$ such that
    $$\phi(S) \leq
    \sqrt{\frac{2\Q^0(f)}{\|f\|_w^2}}.$$
\end{proposition}

\begin{proof}
    Without loss of generality,
    suppose $-1\leq f_u \leq 1$ for all $u\in V$
    since we can scale $f$ if not.
    Let $t$ be a random variable that
    is uniformly distributed in $(0,1]$.
    Define $S_t:=\{u\in V: f_u^2\geq t\}$.
    Then $S_t\subseteq \supp(f)$ by definition.
    We consider the expected value of $w(S_t)$ and $w(\partial S_t)$.
    \begin{align}
    \begin{split}
    \E[w(S_t)]=\sum_{u\in V}w_u\Pr[u\in S_t]
    =\sum_{u\in V}w_u\Pr[f_u^2\geq t]
    =\sum_{u\in V}w_uf_u^2.
    \end{split}
    \nonumber
    \end{align}
    \begin{align}
    \begin{split}
    &\E\left[w(\partial S_t)\right]
    =\sum_{e\in E}w_e\Pr\left[e\in \partial S_t\right]
    =\sum_{e\in E}w_e\Pr\left[
    \min_{i\in e}f_i^2<t\leq\max_{s\in e}f_s^2
    \right]
    \\
    =&\sum_{e\in E}w_e\left(\max_{s\in e}f_s^2-\min_{i\in e}f_i^2\right)
    =\sum_{e\in E}w_e\max_{s,i\in e}\left(f_s^2-f_i^2\right)
    =\sum_{e\in E}w_e\max_{s,i\in e}\left(\Abs{f_s}^2-\Abs{f_i}^2\right)
    \\
    =&\sum_{e\in E}w_e\max_{s,i\in e}
    \left(\Abs{f_s}-\Abs{f_i}\right)\left(\Abs{f_s}+\Abs{f_i}\right)
    \leq\sum_{e\in E}w_e\max_{s,i\in e}\left(\Abs{f_s}-\Abs{f_i}\right)
    \max_{u,v\in e}\left(\Abs{f_u}+\Abs{f_v}\right)
    \\
    \leq&\sum_{e\in E}w_e\max_{s,i\in e}\left(f_s-f_i\right) \cdot
    \max_{u,v\in e}\left(\Abs{f_u}+\Abs{f_v}\right)
    \\
    &\mbox{(By Cauchy--Schwarz inequality)}
    \\
    \leq&
    \sqrt{\sum_{e\in E}w_e
    \max_{s,i\in e}\left(f_s-f_i\right)^2} \cdot
    \sqrt{\sum_{e\in E}w_e
    \max_{u,v\in e}\left(\Abs{f_u}+\Abs{f_v}\right)^2}
    \\
    \leq&\sqrt{\Q^0(f)} \cdot \sqrt{\sum_{e\in E}w_e
    \max_{u,v\in e}2\left(f_u^2+f_v^2\right)} 
    \leq\sqrt{\Q^0(f)} \cdot \sqrt{\sum_{e\in E}w_e \cdot
    2\sum_{u\in e}f_u^2}
    \\
    =&\sqrt{\Q^0(f)} \cdot \sqrt{2\sum_{u\in V}w_uf_u^2}
    =\sqrt{\frac{2\Q^0(f)}{\|f\|_w^2}}
    \cdot \sum_{u\in V}w_uf_u^2.
    \end{split}
    \nonumber
    \end{align}
    Thus,
    $\frac{\E\left[w(\partial S_t)\right]}{\E\left[w(S_t)\right]}
    \leq \sqrt{\frac{2\Q^0(f)}{\|f\|_w^2}}$,
    which means
    $\E\left[w(\partial S_t)
    -\sqrt{\frac{2\Q^0(f)}{\|f\|_w^2}}\cdot w(S_t)\right]
    \leq 0$.
    
    Therefore, there exists a $t$ such that
    the induced $S_t$ satisfies 
    $\phi(S_t)=
    \frac{w(\partial S_t)}{w(S_t)}
    \leq \sqrt{\frac{2\Q^0(f)}{\|f\|_w^2}}$.
\qed
\end{proof}

For any $g \in \R^{V}$,
define $g_+,g_- \in \R^{V}$ such that
$g_+(v)=\max\{0,g(v)\}$ and $g_-(v)=\min\{0,g(v)\}$ for $v\in V$.

\begin{proposition}
    \label{prop:find_hyper_general_cut}
    Given an edge-weighted hypergraph $H = (V,E,w)$ and a non-zero vector $f \in \R^{V}$ such that $f \perp_w \one$, 
    let $g=f-c\one$  where $c$ is a constant such that
    $w(\supp(g_+))\leq \frac{w(V)}{2}$
    and $w(\supp(g_-))\leq \frac{w(V)}{2}$.
    Then we have the following:
    \begin{compactitem}
        \item[1.] $\|g\|_w^2\geq \|f\|_w^2$,
        \item[2.] $\Q^0(g) = \Q^0(f)$.
        \item[3.] $\min\left\{
        \frac{\Q^0(g_+)}{\|g_+\|_w^2},
        \frac{\Q^0(g_-)}{\|g_-\|_w^2}
        \right\}
        \leq \frac{\Q^0(g)}{\|g\|_w^2}
        \leq \frac{\Q^0(f)}{\|f\|_w^2}$.
    \end{compactitem}
\end{proposition}

\begin{proof}
    For the first statement, let
    $h(c)=\|g\|_w^2=\sum_{u\in V}w_ug(u)^2=\sum_{u\in V}w_u(f_u-c)^2$.
    Then we have $h'(c)=\sum_{u \in V}(-2w_uf_u+2cw_u)=\sum_{u \in V}2cw_u$,
    since $f \perp_w \one$.
    We also have $h''(c)=\sum_{u \in V}2w_u>0$.
    Thus, $h(c)$ is minimized when $h'(c)=0$, i.e. $c=0$.
    Then $\|g\|_w^2=h(c)\geq h(0)
    =\sum_{u\in V}w_uf_u^2=\|f\|_w^2$.
    
    For the second statement, we have
    
    $\Q^0(g)
    =\sum_{e\in E}w_e\max_{u,v\in e}[g(u)-g(v)]^2
    =\sum_{e\in E}w_e\max_{u,v\in e}[f(u)-f(v)]^2
    =\Q^0(f).$
    
    For the third statement, notice that
    $\|g\|_w^2=\sum_{u\in V}w_ug(u)^2
    =\sum_{u\in V}w_ug_+(u)^2+\sum_{u\in V}w_ug_-(u)^2
    =\|g_+\|_w^2+\|g_-\|_w^2
    $.
    
    We claim that $\Q^0(g)\geq \Q^0(g_+)+\Q^0(g_-)$ since
    for all $e\in E$,
    
    $\max_{u,v\in e}[g(u)-g(v)]^2
    \geq \max_{u,v\in e}[g_+(u)-g_+(v)]^2
    +\max_{u,v\in e}[g_-(u)-g_-(v)]^2$
    
    \noindent by considering the following two cases:
    \begin{compactitem}
        \item[1.] $g(u)$ and $g(u)$ have the same sign.
        Then it must be either $\forall u\in e, g(u)\geq 0$
        or $\forall u\in e, g(u)\leq 0$.
        Thus, one of $\max_{u,v\in e}[g_+(u)-g_+(v)]^2$
        and $\max_{u,v\in e}[g_-(u)-g_-(v)]^2$ equals zero.
        Then $\max_{u,v\in e}[g(u)-g(v)]^2
        = \max_{u,v\in e}[g_+(u)-g_+(v)]^2
        +\max_{u,v\in e}[g_-(u)-g_-(v)]^2$.
        
        \item[2.] $g(u)$ and $g(u)$ have the opposite signs.
        With out loss of generality, we assume $g(u)>g(v)$.
        Then we have $g_+(u)>0$, $g_-(v)<0$
        and $g_+(v)=g_-(u)=0$.
        Thus, 
        \begin{align}
        \begin{split}
        \max_{u,v\in e}[g(u)-g(v)]^2
        &=\max_{u,v\in e}[g(u)^2-2g(u)g(v)+g(v)^2]
        \\
        &\geq \max_{u\in e:g(u)\geq 0}g(u)^2
        +\max_{v\in e:g(v)\leq 0}g(v)^2
        \\
        &=\max_{u\in e}g_+(u)^2
        +\max_{v\in e}g_-(v)^2
        \\
        &=\max_{u,v\in e}[g_+(u)-g_+(v)]^2
        +\max_{u,v\in e}[g_-(u)-g_-(v)]^2.
        \end{split}
        \nonumber
        \end{align}
    \end{compactitem}
Then by the first and second statements, we have
$$\frac{\Q^0(f)}{\|f\|_w^2}
\geq \frac{\Q^0(g)}{\|g\|_w^2}
\geq \frac{\Q^0(g_+)+\Q^0(g_-)}{
\|g_+\|_w^2+\|g_-\|_w^2}
\geq \min\left\{
\frac{\Q^0(g_+)}{\|g_+\|_w^2},
\frac{\Q^0(g_-)}{\|g_-\|_w^2}
\right\}.
$$
\qed
\end{proof}

Now we can prove Theorem~\ref{thm:hyper-cheeger}.

\begin{proofof}{Theorem~\ref{thm:hyper-cheeger}}~
    By Proposition~\ref{prop:find_hyper_cut},
    there exist $S_{t+}\subseteq \supp(g_+)$
    and $S_{t-}\subseteq \supp(g_-)$
    such that for any non-zero vector
    $f\in \R^V$ satisfying $f \perp_w \one$:
    \begin{align}
    \begin{split}
    \min\left\{\phi(S_{t+}),\phi(S_{t-})\right\}
    \leq\min\left\{
    \sqrt{\frac{2\Q^0(g_+)}{\|g_+\|_w^2}},
    \sqrt{\frac{2\Q^0(g_-)}{\|g_-\|_w^2}}
    \right\}
    \leq \sqrt{\frac{2\Q^0(f)}{\|f\|_w^2}}
    \end{split}
    \nonumber
    \end{align}
    where the last inequality follows from the third statement of Proposition~\ref{prop:find_hyper_general_cut}.
    Moreover, $w(S_{t+})\leq\frac{w(V)}{2}$ and $w(S_{t-})\leq\frac{w(V)}{2}$
    by Proposition~\ref{prop:find_hyper_general_cut}.
    Thus, for any non-zero vector
    $f\in \R^V$ satisfying $f \perp_w \one$:
    $$
    \phi_H
    \leq \min\left\{\phi(S_{t+}),\phi(S_{t-})\right\}
    \leq \sqrt{\frac{2\Q^0(f)}{\|f\|_w^2}}.
    $$
    Then we have $\phi_H \leq
    \min_{\zero\not=f \perp_w \one}
    \sqrt{\frac{2\Q^0(f)}{\|f\|_w^2}}=\sqrt{2\gamma^0_2}$.
\end{proofof}

\end{document}